\documentclass[final]{article}

\usepackage{microtype}
\usepackage{amsmath}
\usepackage{graphicx}
\usepackage{xcolor }
\usepackage{soul}
\usepackage{algorithmic}
\usepackage{tabto}
\usepackage{caption}
\usepackage{subcaption}
\newcommand{\remove}[1]{}
\usepackage{graphicx}
\usepackage{enumerate}
\usepackage[linesnumbered,ruled,vlined,resetcount,noend]{algorithm2e}
\usepackage{verbatim}
\usepackage{amsthm}

\newtheorem{theorem}{Theorem}
\newtheorem{lemma}{Lemma}
\newtheorem{claim}{Claim}
\newtheorem{corollary}{Corollary}

\newcommand{\facilityset}{\mathcal{F}}	
	
\newcommand{\clientset}{\mathcal{C}}			

\newcommand{\capacityU}{\mathcal{U}}				

\newcommand{\dist}[2]{c(#1,~#2)}

\newcommand{\etal}{et al.}

\newcommand{\B}{\mathcal{L}}

\newcommand{\opt}[1]{LP_{opt}}

\newcommand{\sstar}[1]{\mathcal{S}_{#1}}

\newcommand{\etainv}[1]{ \eta^{-1}(#1)}
\newcommand{\res}[1]{reserved(#1)}
\newcommand{\sigmatwoinv}[1]{\sigma^{-1}_2(#1)}
\newcommand{\sigmaoneinv}[1]{\sigma^{-1}_1(#1)}
\newcommand{\sigmaprimeoneinv}[1]{\hat{\sigma}^{-1}_1(#1)}
\newcommand{\sigmaprimeone}{\hat{\sigma}_1}
\newcommand{\settled}{\clientset_s}
\newcommand{\unsettled}{\clientset_u}

\date{}
\begin{document}

\title{Constant factor approximations for Lower and Upper bounded Clusterings}
\maketitle          
	\begin{center}
		\author{Neelima Gupta$^1$, }
		\author{Sapna Grover$^2$ and}
		\author{Rajni Dabas$^3$  }\\
		Department of Computer Science, University of Delhi, India.\\
	\end{center}
	\begin{enumerate}
		\item {	\texttt{ngupta@cs.du.ac.in}} 
		\item {\texttt{sgrover@cs.du.ac.in, sapna.grover5@gmail.com}}
		\item {	\texttt{rajni@cs.du.ac.in}}
	\end{enumerate}

\begin{abstract}

Clustering is one of the most fundamental problem in Machine Learning. Researchers in the field often require a lower bound on the size
of the clusters to maintain anonymity and upper bound for the ease of analysis. Specifying an optimal cluster size is a problem often faced by scientists.  
In this paper, we present a framework to obtain constant factor approximations for some prominent clustering objectives, with lower and upper bounds on cluster size.
This enables scientists to give an approximate cluster size by specifying the lower and the upper bounds for it. Our results preserve the lower bounds but may violate the upper bound a little.

We also reduce the violation in upper bounds for a special case when the gap between the lower and upper bounds is not too small.\\

\textbf{Keywords:}  Facility Location, $k$-Median, $k$-Center, Lower Bound, Upper Bound, Approximation.

\end{abstract}

\section{Introduction}

{\em Clustering} is one of the most fundamental problem in Machine Learning: given a set of data points, we wish to group them into clusters so that the points within a cluster are most similar to each other and those in different clusters are most dissimilar. Various objective functions have been used in literature to capture the notion of similarity. Three objectives that have been extensively studied in combinatorial optmization are captured in  $k$-Center (kC)/$k$-Median (kM) and facility location (FL) problems. In the $k$-Center (/$k$-Median) problem, we wish to identify a set of $k$ centers and assign points to it so as to minimise the maximum distance(dissimilarity) of any point from its assigned center(/minimise the sum of distances of the points from the assigned center). In FL, instead of a bound on the number of centers to open, we have center(/facility) opening cost and we wish to form clusters so that sum of facility opening cost and the sum of distances of the points(/clients) from the assigned center is minimised. 
$k$-Facility Location problem (kFL) is a common generalization of kM and FL, where-in we have a bound $k$ on the maximum number of facilities to be opened and facility opening costs associated with each facility. And, the objective now is to minimize the total cost of opening the selected facilities and assigning clients to them. 

The problems are well known to be NP-hard. Several approximation results have been developed for these problems in the basic version. 
However, constraints occur naturally in clustering.  Adding constraints makes the problem harder. 
One such constraint requires every center to have a minimum number of data points assigned to it. An application requiring lower bound on the cluster size is data privacy~\cite{Samarati2001} that demands every point in the cluster to be "alike" and indistinguishable from each other. In the FL inspired clustering, lower bounds are required to ensure profitability.

Another constraint imposes an upper bound on the maximum number of data points that can be assigned to a center. In the FL inspired clustering, capacities occur naturally on the facilities. In case of cluster analysis, researchers often do not want the clusters to be too big for the ease of analysis. 

In this paper, we study the problems with both the lower as well as the upper bounds  on the minimum and the maximum number, respectively, of data points that are assigned to a center. 
Researchers often face the problem of determining the appropriate size of clusters and, several heuristics are used to estimate the same~\cite{Jain2010,GomezGN12}. In this paper, we handle this situation somewhat giving them an opportunity to give a rough estimate of the size in terms of lower and upper bound on the size. Heuristics provide no performance guarantee whereas our solution generates clusters that respect the lower bounds and violate the upper bounds a little. In FL inspired clustering, limitation of capacities and the requirement of scale for profitability occur together in several applications like market-place, transportation problem etc. We present our results when one of the bounds is uniform.  As upper bound has been popularly called as capacity in the literature, we will use the two terms inter-changeably in the paper.
In the $k$-center problem, the set of data points is same as the set of centers i.e., the centers can be opened at the data points. $k$- supplier (kS) is a generalization of the $k$- center problem in which the two sets are different. Our results are applicable to the $k$-supplier problem.



\begin{theorem} \label{LBUBkFL1Theorem1}
Given an $\alpha$-approximation for Upper bounded $k$-Median problem (UkM)/$k$-Facility Location problem (UkFL)/Facility Location problem (UFL)/$k$-Supplier problem (UkS)/$k$-Center problem (UkC) violating the upper bound by a factor of $\beta$ and a $\gamma$-approximation for uniform Lower bounded Facility Location problem (LFL)/$k$-Supplier problem (LkS)/$k$-Center problem (LkC), an $O(\alpha + \gamma)$-approximation can be obtained for Lower and  Upper bounded $k$-Median problem (LUkM)/$k$-Facility Location problem (LUkFL)/Facility Location problem (LUFL)/$k$-Supplier problem (LUkS)/$k$-Center problem (LUkC) with uniform lower bounds in polynomial time that violates the upper bound by a factor of $\beta + 1$.
\end{theorem}


\begin{theorem} \label{LBUBkFLTheorem}
Given an $\alpha$-approximation for uniform Upper bounded $k$-Median problem (UkM)/$k$-Facility Location problem (UkFL)/Facility Location problem (UFL)/$k$-Supplier problem (UkS)/$k$-Center problem (UkC) violating the upper bound by a factor of $\beta$ and a $\gamma$-approximation for Lower bounded Facility Location problem (LFL)/$k$-Supplier problem (LkS)/$k$-Center problem (LkC), an  $O(\alpha + \gamma)$-approximation can be obtained for Lower and Upper bounded $k$-Median problem (LUkM)/$k$-Facility Location problem (LUkFL)/Facility Location problem (LUFL)/$k$-Supplier problem (LUkS)/$k$-Center problem (LUkC) with uniform upper bounds in polynomial time that violates the upper bound by a factor of $\beta + 1$.
\end{theorem}

Several studies have been devoted to the lower bounds 
and the upper bounds 
separately for the problems. The only works that deal with both the bounds together is due to Gupta \etal~\cite{GroverGD21_LBUBFL_Cocoon} and Friggstad \etal~\cite{friggstad_et_al_LBUFL}  for FL and,  Ding \etal~\cite{DingWADS2017} and, R{\"{o}}sner and Schmidt~\cite{RosnerICALP2018} for kS and kC.
 Friggstad \etal~\cite{friggstad_et_al_LBUFL} violates  bounds on both the sides for non-uniform lower and uniform upper bounds whereas Gupta \etal~\cite{GroverGD21_LBUBFL_Cocoon} violates only the upper bound by a factor of $5/2$ when both the bounds are uniform. For LUkS/LUkC,  Ding \etal~\cite{DingWADS2017} and, R{\"{o}}sner and Schmidt~\cite{RosnerICALP2018} independently gave true constant factor approximations for uniform lower bounds and general upper bounds.

 Thus, we present first approximations for LUkM and LUkFL. For LUFL, we reduce the $5/2$ factor violation in upper bounds obtained by Gupta \etal~\cite{GroverGD21_LBUBFL_Cocoon} to $2$. Our result for LUkC and LUkS are the first results with general lower bounds and uniform upper bounds. 


\begin{table}
    \centering
    \begin{tabular}{|l|c|c|c|c|c|c|}
    \hline
          \textbf{Problem} & \textbf{$\alpha$} & \textbf{$\beta$} & \textbf{$\gamma$}  & \textbf{Factor} & \textbf{Capacity violation} \\ \hline
          
           LUkM & $O(1/\epsilon^5)$~\cite{Demirci2016} & $(1+\epsilon)$ & $82.6$~\cite{Ahmadian_LBFL} &  $O(1/\epsilon^5)$ & $(2+\epsilon)$ \\ \hline
           
            LUFL & $5$~\cite{Bansal} & $1$ & $82.6$~\cite{Ahmadian_LBFL} & O(1) & 2 \\ \hline
         
            LUkS & $3$~\cite{AnMP2015} & $1$ & $3$~\cite{Ahmadian2016}  & O(1) & 2 \\ \hline
            
            LUkC & $2$~\cite{AnMP2015} & $1$ & $2$~\cite{Aggarwal2010}  & O(1) & 2 \\ \hline
            
    \end{tabular}
    \caption{Our results for uniform lower bound.}
    \label{tab1.1}
\end{table}

\begin{table}
    \centering
    \begin{tabular}{|l|c|c|c|c|c|}
    \hline
          \textbf{Problem} & \textbf{$\alpha$} & \textbf{$\beta$} & \textbf{$\gamma$} & \textbf{Factor} & \textbf{Capacity violation} \\ \hline
          
           LUkM & $O(1/\epsilon^2)$~\cite{ByrkaRybicki2015} & $(1+\epsilon)$ & O(1)~\cite{Li_NonUnifLBFL} & $O(1/\epsilon^2)$ & $(2+\epsilon)$ \\ \hline
           
            LUkFL & $O(1/\epsilon^2)$~\cite{GroverGKP18} & $(2+\epsilon)$ & O(1)~\cite{Li_NonUnifLBFL} & $O(1/\epsilon^2)$& $(3+\epsilon)$ \\ \hline
                     
            LUFL & $3$~\cite{mathp} & $1$  & O(1)~\cite{Li_NonUnifLBFL} & O(1) & 2 \\ \hline
            
            LUkS & $3$~\cite{AnMP2015} & $1$ & $3$~\cite{Ahmadian2016}  & O(1) & 2 \\ \hline
            
            LUkC & $2$~\cite{k-centerKhuller2000} & $1$ & $3$~\cite{Ahmadian2016}  & O(1) & 2 \\ \hline

    \end{tabular}
    \caption{Our results for uniform upper bound.}
    \label{tab1}
\end{table}

Tables~\ref{tab1.1} and \ref{tab1} summarize our results after applying the current best known results for the underlying problems.
Though no results are known for UkFL with general upper bounds and general facility opening costs, to plug into the result of Theorem \ref{LBUBkFL1Theorem1}, the result will be useful in case we get one in future.
As a special case of LUkFL, we obtain first true constant factor approximation for kFL with lower bounds (LkFL). The only result known for LkFL by Han et al.~\cite{Han_LBkM} for general lower bounds, violates the lower bounds. The following corollary follows from the theorem:

\begin{corollary} \label{LBUBkFLCorollary}
There is a polynomial time algorithm that approximates Lower bounded $k$-Facility Location within a constant factor.
\end{corollary}

Next, we improve upon the violation in the upper bound for a particular scenario when the gap between the lower and the upper bounds is not too small (a reasonable scenario to occur in real applications). In particular, if $\B_i$ and $\capacityU_i$ represent the lower and the upper bounds respectively of a facility $i$, we present the following results when $2\B_{i} \leq \capacityU_i \forall i \in \facilityset$ and either of the bounds is uniform.

\begin{theorem} \label{FinalResultLBUBkFL2L11}
For $2\B_{i} \leq \capacityU_i \forall i \in \facilityset$, given an $\alpha$-approximation for 
Upper bounded $k$-Median problem (UkM)/$k$-Facility Location problem (UkFL)/Facility Location problem (UFL)/$k$-Supplier problem (UkS)/$k$-Center problem (UkC) violating the upper bound by a factor of $\beta$ and a $\gamma$-approximation for uniform Lower bounded Facility Location problem (LFL)/$k$-Supplier problem (LkS)/$k$-Center problem (LkC), an  $(O(\alpha + \gamma)/\epsilon)$-approximation can be obtained for Lower and Upper bounded $k$-Median problem (LUkM)/$k$-Facility Location problem (LUkFL)/Facility Location problem (LUFL)/$k$-Supplier problem (LUkS)/$k$-Center problem (LUkC) with uniform lower bounds in polynomial time that violates the upper bound by a factor of $\beta + \epsilon$ for a fixed $\epsilon > 0$. 
\end{theorem}

\begin{theorem} \label{FinalResultLBUBkFL2L1}
For $2\B_{i} \leq \capacityU_i \forall i \in \facilityset$, given an $\alpha$-approximation for uniform Upper bounded $k$-Median problem (UkM)/$k$-Facility Location problem (UkFL)/Facility Location problem (UFL)/$k$-Supplier problem (UkS)/$k$-Center problem (UkC) violating the upper bound by a factor of $\beta$ and a $\gamma$-approximation for Lower bounded Facility Location problem (LFL)/$k$-Supplier problem (LkS)/$k$-Center problem (LkC), an  $(O(\alpha + \gamma)/\epsilon)$-approximation can be obtained for Lower and Upper bounded $k$-Median problem (LUkM)/$k$-Facility Location problem (LUkFL)/Facility Location problem (LUFL)/$k$-Supplier problem (LUkS)/$k$-Center problem (LUkC) with uniform upper bounds  in polynomial time that violates the upper bound by a factor of $\beta + \epsilon$ for a fixed $\epsilon > 0$. 
\end{theorem}



Tables~\ref{tab2} and \ref{tab21} present our results after applying Theorem~\ref{FinalResultLBUBkFL2L11} and \ref{FinalResultLBUBkFL2L1} to different problems for  $2\B_{i} \leq \capacityU_i \forall i \in \facilityset$. 

\begin{table}
    \centering
    \begin{tabular}{|l|c|c|c|c|c|c|}
    \hline
          \textbf{Problem} & \textbf{$\alpha$} & \textbf{$\beta$} &  \textbf{$\gamma$} & \textbf{Factor} & \textbf{Capacity Violation} \\\hline
          
           LUkM & $O(1/\epsilon^5)$~\cite{Demirci2016} & $(1+\epsilon)$ & $82.6$~\cite{Ahmadian_LBFL} &  $O(1/\epsilon^6)$ & $(1+\epsilon)$ \\ \hline
           
            LUFL & $5$~\cite{Bansal} & $1$ & $82.6$~\cite{Ahmadian_LBFL} & $O(1/\epsilon)$ & $(1+\epsilon)$ \\ \hline
            
            LUkS & $3$~\cite{AnMP2015} & $1$ & $3$~\cite{Ahmadian2016}  & O(1) &  $(1+\epsilon)$  \\ \hline
            
            LUkC & $2$~\cite{AnMP2015} & $1$ & $2$~\cite{Aggarwal2010}  & O(1) &  $(1+\epsilon)$  \\ \hline
         
    \end{tabular}
    \caption{Our results for uniform lower bounds when $2\B_i \leq \capacityU_i~\forall i \in \facilityset$}.
    \label{tab2}
\end{table}

\begin{table}
    \centering
    \begin{tabular}{|l|c|c|c|c|c|c|}
    \hline
          \textbf{Problem} & \textbf{$\alpha$} & \textbf{$\beta$} & \textbf{$\gamma$} & \textbf{Factor} & \textbf{Capacity Violation} \\ \hline
          
          LUkFL & $O(1/\epsilon^2)$~\cite{GroverGKP18} & $(2+\epsilon)$ & O(1)~\cite{Li_NonUnifLBFL} & $O(1/\epsilon^3)$& $(2+\epsilon)$ \\ \hline
          
           LUkM & $O(1/\epsilon^2)$~\cite{ByrkaRybicki2015} & $(1+\epsilon)$ & O(1)~\cite{Li_NonUnifLBFL} &  $O(1/\epsilon^3)$ & $(1+\epsilon)$ \\ \hline
           
           LUFL & $3$~\cite{mathp} & $1$ & O(1)~\cite{Li_NonUnifLBFL} & $O(1/\epsilon)$ & $(1+\epsilon)$ \\ \hline
            
           LUkS & $3$~\cite{AnMP2015} & $1$ & $3$~\cite{Ahmadian2016} & O(1) & $(1+\epsilon)$  \\ \hline
            
           LUkC & $2$~\cite{k-centerKhuller2000} & $1$ & $3$~\cite{Ahmadian2016} & O(1) & $(1+\epsilon)$  \\ \hline
         
    \end{tabular}
    \caption{Our results for uniform upper bounds when $2\B_i \leq \capacityU_i~\forall i \in \facilityset$.}
    \label{tab21}
\end{table}

\subsection{Related Work}

The classic $k$-center problem has been approximated with constant factor ratio ($2$) by Gonzalez~\cite{Gonzalez1985} and,  Hochbaum and Shmoys~\cite{HochbaumShmoys1986}. The ratio is tight, i.e., it is NP-hard to approximate it with a factor less than $2$ unless P=NP~\cite{Hsu1979}. Hochbaum and Shmoys~\cite{HochbaumShmoys1986} also approximated $k$-supplier to $3$ factor, which is also tight. Upper bounded $k$-center problem  was first studied by Bar-Ilan \etal~\cite{BarIlan1993} in 1993 for uniform capacities. The authors gave a $10$-factor approximation algorithm. The factor was subsequently improved to $6$ by Khuller and Sussmann~\cite{k-centerKhuller2000}. For non-uniform capacities, Cygan \etal~\cite{CyganFOCS2012} gave the first constant factor approximation algorithm. Their factor was large, which was later improved to $9$ by An \etal~\cite{AnMP2015}. An $11$ factor approximation algorithm for UkS problem was given by An \etal~\cite{AnMP2015}. 
Aggarwal \etal~\cite{Aggarwal2010} introduced and gave the first true constant ($2$) factor approximation for $k$-center with uniform lower bounds. For $k$-supplier and $k$-center  with non-uniform lower bounds, Ahmadian and Swamy~\cite{Ahmadian2016} gave the first true constant factor ($3$) approximation algorithm.

$k$-Median has been extensively studied in past~\cite{Archer03lagrangianrelaxation,arya,ByrkaPRST14, Charikar,Charikar:1999,Charikaricalp2012,Jain:2001, li2013approximating} with the best approximation ratio of $2.611 + \epsilon$ by Byrka \etal ~\cite{ByrkaPRST14}.  Several constant factor approximations are known for upper bounded kM using LP rounding~\cite{capkmGijswijtL2013,capkmByrkaFRS2013,Charikar,Charikar:1999,ChuzhoyR05,GroverGKP18,KPR,capkmshanfeili2014} that violate upper bounds or the cardinality constraint by a factor of $2$ or more.  
Barrier of $2$ was broken to obtain $(1+\epsilon)$ violation in upper bounds  for uniform and non-uniform bounds by Byrka \etal~\cite{ByrkaRybicki2015} and  Demirci \etal~\cite{Demirci2016} respectively by strengthening the LP.	Li gave constant factor approximations using at most $(1 + \epsilon) k$ facilities for uniform~\cite{capkmshili2014} as well as non-uniform~\cite{Lisoda2016} upper bounds.
 For kM with lower bounds (LkM), Guo et al. \cite{Guo_LBkM}, Han \etal~\cite{Han_AIMM} and, Arutyunova and Schmidt \cite{arutyunova_LBkM} gave true constant factor approximations.
    For LkFL, Han et al. \cite{Han_LBkM} gave a constant factor approximation that violates the lower bounds.

A wide range of approximation algorithms  have been developed for the uncapacitated facility location problem~\cite{arya,Byrka07,charikar2005improved,Chudak98,ChudakS03,guha1999greedy,Jain:2001,KPR,Li13,mahdian2001greedy,mahdian_1.52,PalT03,Shmoys,sviridenko2006improved,Zhang2007126}. Li~\cite{Li13} gave a $1.488$ factor algorithm almost closing the gap between the best known and best possible of $1.463$ by Guha \etal ~\cite{guha1999greedy} for the problem. 
UFL has been extensively studied in past~\cite{Anfocs2014,capkmByrkaFRS2013,GroverGKP18,LeviSS12,Shmoys}. For uniform bounds, Shmoys et al.~\cite{Shmoys} gave the first constant factor($7$) approximation violating the upper bounds by a factor of $7/2$ using LP rounding technique.  An $O(1/\epsilon^2)$-approximate solution, with $(2 + \epsilon)$ violation in upper bounds was given by Byrka et al.~\cite{capkmByrkaFRS2013} as a special case of UkFL. Grover et al.~\cite{GroverGKP18} further reduced the violation to $(1 + \epsilon)$.  For non-uniform capacities, Levi et al. \cite{LeviSS12} gave a 5-factor approximation algorithm with uniform facility opening costs. An et al.~\cite{Anfocs2014} strengthened the LP to obtain a true constant factor approximation for the problem. 
Local search technique has been particularly successful in dealing with UFL   for uniform capacities~\cite{mathp,ChudakW99,KPR} and non-uniform capacities~\cite{Bansal,mahdian_universal,paltree,zhangchenye}. The best known results are $3$-factor and $5$-factor approximation by Aggarwal et al.~\cite{mathp} and Bansal et al.~\cite{Bansal} for uniform and non-uniform capacities respectively.
Facility location with lower bounds was first introduced by Karger and Minkoff \cite{Minkoff_LBFL} and, Guha et al.~\cite{Guha_LBFL} independently in 2000.
Both the works violate the lower bounds.
Svitkina \cite{Zoya_LBFL} gave the first true constant factor ($448$) approximation  with uniform lower bounds, which was subsequently improved to $82.6$ by Ahmadian and Swamy \cite{Ahmadian_LBFL}. For non-uniform lower bounds, Li \cite{Li_NonUnifLBFL} gave the first true constant ($4000$) factor approximation for the problem.

      Several constant factor approximation algorithms have been developed for the uncapacitated $k$-facility location problem~\cite{Charikar:1999,Jain,jain2002new,Jain:2001,Zhang2007126} with the best approximation ratio being $ 2 + \sqrt{3} +\epsilon$ due to Zhang~\cite{Zhang2007126}.
    For UkFL problem, Aardal \etal ~\cite{capkmGijswijtL2013} gave a constant factor approximation  for the problem with uniform opening costs, using at most $2k$ facilities for non-uniform capacities and $2k - 1$ facilities for uniform capacities. Byrka \etal~\cite{capkmByrkaFRS2013} gave an $O(1/\epsilon^2)$ factor approximation for uniform capacities violating the capacities by a factor of $(2 + \epsilon)$. 
    Grover \etal~\cite{GroverGKP18} also gave two constant factor results on the problem with a trade-off between the capacity and the cardinality violations: ($i$) $(1+\epsilon)$ violation in capacity with $2/(1+\epsilon)$\footnote{A careful analysis of the result shows that the violation is $2/(1+\epsilon)$} factor cardinality violation and ($ii$) $(2+\epsilon)$ violation in capacity and no cardinality violation. To the best of our knowledge, no result is known for the problem with non-uniform capacities and general opening costs
     
Results for bounds on both the sides were first obtained by Friggstad et al.~\cite{friggstad_et_al_LBUFL} for facility location problem. The paper considers general lower bounds and uniform upper bounds. They gave the first constant factor approximation using LP rounding techniques, violating both, the lower and the upper bounds.  Gupta \etal~\cite{GroverGD21_LBUBFL_Cocoon} gave the first constant factor approximation when both the bounds are uniform, violating the upper bounds by a factor of $5/2$ while respecting the lower bounds.  Constant factor results were obtained for LUkC and LUkS with uniform lower bounds and general upper bounds independently by Ding \etal~\cite{DingWADS2017} and, R{\"{o}}sner and Schmidt \cite{RosnerICALP2018}.
     

\subsection{High Level Idea}
{
Let $I$ be an input instance.
For LUkM/LUkFL/LUFL, we create two instances: $I_1$ of LFL and $I_2$ of UkM/UkFL/UFL
from $I$ by dropping the upper bounds and the cardinality constraint for $I_1$ and, dropping the lower bound for $I_2$. For LUkS/LUkC, LFL does not work as the component of service cost in kS/kC is different from the one in FL.
Thus, for LUkS/LUkC, we create two instances: $I_1$ of LkS/LkC and $I_2$ of UkS/UkC from $I$ in a similar way.
Since, any solution to $I$ is feasible for $I_1$ and $I_2$, costs of $I_1$ and $I_2$ are bounded. We obtain approximate solutions $AS_1$ and $AS_2$ to $I_1$ and $I_2$ respectively and, combine them to obtain a solution for $I$.
}
If the solution of $I_2$
violates the upper bound by a factor of 
{$\beta$,} then our solution violates it by a factor of 
{$\beta+1$.}
If there were no upper bounds, we could do the following: For every facility $i$ opened in $AS_1$, we open its nearest facility $\eta(i)$ in $AS_2$ and assign its clients to $\eta(i)$. The cost of doing this can be bounded. However, we can not do this in the presence of upper bounds as $i$ could be serving any number of clients in $AS_1$ and hence $\eta(i)$ may be assigned arbitrary number of clients. To take care of the upper bounds, we do the following:
for every facility $i'$ in $AS_2$, we identify its nearest facility $\eta(i')$ in $AS_1$. We treat all such facilities $i'$'s, whose nearest facility $\eta(i')$ is same, along with $\eta(i')$, together and open at most all but one of them. Some of the facilities in $AS_2$ that violate the lower bound are closed transferring their clients to other facilities. When we arrive at a facility that has received sufficient number of clients, we open it and assign all the clients to it. Clearly such a facility receives no more clients than
$\beta+1$ times its capacity.


\subsection{Organisation of the Paper}
Section \ref{LBUBkM} presents our framework for the LU(lower and upper bounded) variants of the various clustering problems under consideration. 
In Section \ref{2LU}, we improve our results for the special case when $2\B_{i} \leq \capacityU_{i}~\forall i \in \facilityset$.
We finally conclude in Section~\ref{cncl} giving scope of future work.

\section{Framework for Lower and Upper Bounded Clustering}
\label{LBUBkM}

In this section, we present our framework for lower and upper bounded variants of the clustering problems under consideration.

\noindent
{\bf Problem Definitions:} Given a set $\clientset$ of clients, a set $\facilityset$ of facilities, opening cost $f_i$ associated with facility $i$, metric cost $\dist{i}{j}$ (called the {\em service cost}) of serving a client $j$ from a facility $i$,  lower bounds $\B_i$ and upper bounds $\capacityU_i$ on the minimum and the maximum number of clients,  respectively,  an open facility $i$ should serve, cardinality bound $k$ on the maximum number of facilities that can be opened, the objective is to open a subset $\facilityset' \subseteq \facilityset$ and compute an assignment function $\sigma : \clientset \rightarrow \facilityset'$ (where $\sigma(j)$ denotes the facility that serves client $j$ in the solution)
such that $|\facilityset'| \leq k$ and $ \B_i \leq |\sigma^{-1}(i)| \leq  \capacityU_i ~\forall i \in \facilityset'$. We consider the following variants of the clustering problems:

\begin{itemize}
    \item {\bf Lower and Upper bounded $k$-Median (LUkM)}: $f_i = 0~\forall i \in \facilityset, \facilityset \cap \clientset = \emptyset$. We wish to minimise $\sum_{i \in \facilityset'} \sum_{j \in \clientset} \dist{j}{\sigma(j)}$.
    
    \item {\bf Lower and Upper bounded Facility Location (LUFL)}: $k = |\facilityset|, \facilityset \cap \clientset = \emptyset$. The aim is to minimise $\sum_{i \in \facilityset'} f_i + \sum_{i \in \facilityset'}\sum_{j \in \clientset} \dist{j}{\sigma(j)}$.
    
     \item {\bf Lower and Upper bounded $k$-Facility Location (LUkFL)}: $\facilityset \cap \clientset = \emptyset$. The aim is to minimise $\sum_{i \in \facilityset'} f_i + \sum_{i \in \facilityset'}\sum_{j \in \clientset} \dist{j}{\sigma(j)}$.
     
      \item {\bf Lower and Upper bounded $k$-Center (LUkC)}:  $f_i = 0~\forall i \in \facilityset, \facilityset = \clientset$. The aim is to minimise $max_{j \in \clientset} \dist{j}{\sigma(j)}$.
    
     \item {\bf Lower and Upper bounded $k$-Supplier (LUkS)}:  $f_i = 0~\forall i \in \facilityset, \facilityset \cap \clientset = \emptyset$.  The aim is to minimise $max_{j \in \clientset} \dist{j}{\sigma(j)}$.
\end{itemize}

Algorithm $\ref{LBUBkM_Algo}$ presents the main steps of the framework
where $Cost_\mathcal{I}(\mathcal{S})$ denotes the cost of a solution $\mathcal{S}$ to an instance $\mathcal{I}$. 
Steps $1 - 4$ are self explanatory. We present the main step, step $5$ in detail in Section~\ref{bicriteriakM}.



\begin{algorithm}[H] 
\label{LBUBkM_Algo}
    \setcounter{AlgoLine}{0}
	\SetAlgoLined
    \DontPrintSemicolon  
    \SetKwInOut{Input}{Input}
	\SetKwInOut{Output}{Output}
    \Input{Instance $I$}

    
     For LUkM/LUkFL/LUFL, create an instance $I_1$ of LFL from $I$ by dropping the upper bounds and the cardinality constraint and for LUkS/LUkC, create $I_1$ of LkS/LkC by dropping the upper bounds. Since any solution to $I$ is feasible for $I_1$ as well, we have  $Cost_{I_1}(O_1) \leq Cost_{I}(O)$, where
    $O$ and $O_1$ denote optimal solution to $I$ and $I_1$ respectively.\;

    Create an instance $I_2$ of UkM/UkFL/UFL/UkS/UkC  by dropping the lower bounds from $I$. Since any solution to instance $I$ is feasible for $I_2$ as well, we have $Cost_{I_2}(O_2) \leq Cost_{I}(O)$, where $O_2$ denotes an optimal solution to $I_2$.\;

    Obtain a $\gamma$-approximate solution $AS_1 = (\facilityset_1,\sigma_1)$ to instance $I_1$.\;
    
    Obtain an $\alpha$-approximate solution $AS_2 = (\facilityset_2,\sigma_2)$ to instance $I_2$. Let $\beta$ denote the violation in upper bounds, if any, in $AS_2$.\;

    Combine solutions $AS_1$ and $AS_2$ to obtain a  solution $AS_I= (\facilityset_I, \sigma_I)$ to $I$ such that
    {$Cost_{I}(AS_I) \leq \Delta_1 Cost_{I_1}(AS_1) + \Delta_2 Cost_{I_2}(AS_2) $}
    with $(\beta+1)$ violation in upper bound, where $\Delta_1$ and $\Delta_2$ are positive constants.
	 
	\caption{
	}
	\label{LBUBkM_Algo}
\end{algorithm}

\subsection{Combining solutions $AS_1$ and $AS_2$ to obtain $AS_I$} \label{bicriteriakM}

In this section, we combine solutions  $AS_1=(\facilityset_1, \sigma_1)$ and $AS_2 = (\facilityset_2, \sigma_2)$ to obtain a 
solution $AS_I= (\facilityset_I, \sigma_I)$ to $I$ respecting the lower bounds on facilities and the cardinality constraint, while violating the upper bounds by a factor of $(\beta+1)$. 

 We will open some facilities in $\facilityset_1 \cup \facilityset_2$. Note that
the facilities in $\facilityset_1$ may violate the upper bounds and, the facilities in $\facilityset_2$ may violate the lower bounds. Thus, we need to select which facilities in $\facilityset_1 \cup \facilityset_2$ to open closing the remaining facilities. To facilitate this, we construct a directed graph $G_1$ on the set of facilities in $\facilityset_1 \cup \facilityset_2$. For a facility $i \in \facilityset_2$, let $\eta(i)$ denote the  facility in $\facilityset_1$ nearest to $i$ (assuming that the distances are distinct). Add an edge $(i, \eta(i))$ in the graph.  
In order to avoid self loops when $i = \eta(i)$, we denote the occurrence of  $i$ in $\facilityset_2$ by $i_c$ so that $\eta(i_c) = i$. Thus, we obtain a forest of trees where-in each tree is a {\em star}. Formally, we define a star $\sstar{i}$ to be a collection of nodes in $\{i \cup \eta^{-1}(i)\}$ with $i \in \facilityset_1$ as the {\em center} of the star and $\eta^{-1}(i) \subseteq \facilityset_2$. See Figure \ref{star}. Though every facility in $\facilityset_2$ belongs to some star, there may be facilities in $\facilityset_1$ that do not belong to any star.

\begin{figure}[]
		\begin{center}
	    \begin{tabular}{c}
		\includegraphics[width=80mm,scale=0.65]{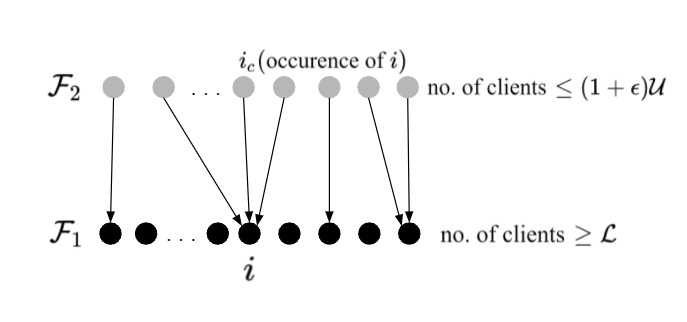}
	    \end{tabular}
		\end{center}
		\caption{
	Graph $G_1$: $I_1$ is an instance of LFL/LkS/LkC and $I_2$ that of UkM/UkFL/UFL/UkS/UkC.
		}
	\label{star}
\end{figure}

We process the stars to decide the set of facilities to open in $\facilityset_1 \cup \facilityset_2$.
Consider a star $\sstar{i}$ centered at facility $i$. Clearly the total assignments on $i$ in $AS_1$ satisfy the lower bound but may violate the upper bound arbitrarily. On the other hand, the total assignments on a facility $i' \in \eta^{-1}(i)$ in $AS_2$ satisfy the upper bound (within $\beta$ factor) but may violate the lower bound arbitrarily. We close some facilities in $ \eta^{-1}(i)$ by transferring their clients (assigned to them in $AS_2$) to other facilities in $\etainv{i}$ (and to $i$, if required) and open those who have got at least $\B$ clients. We may also have to open $i$ in the process. We make sure that upper bound is violated within the claimed bounds at $i$ and the total number of facilities opened in $\sstar{i}$ is at most $|\eta^{-1}(i)|$. 
{ Thus, we open no more than $k$ facilities in all.}

To determine which facilities in $\etainv{i}$ to open and which facilities to close,
we consider the facilities in $\etainv{i}$ in the order of decreasing distance from $i$. Let the order be $y_1, y_2,...,y_l$.
 We collect the clients assigned to them, by $AS_2$, in a bag looking for a facility $t$ at which we would have collected at least $\B_t$ clients. We open such a facility and empty the bag by assigning all the clients in the bag to the facility just opened  and start the process again with the next facility in the order. The problem occurs when at the last facility ($y_l$), in the order, the bag has less than $\B_{y_l}$ clients. In this case, we would like to assign these clients to the star center $i$ making use of the fact that $i$ was assigned at least $\B_i$ clients in $AS_1$. The problem here is that the clients assigned to $i$ in $AS_1$ might have been assigned to the facilities in $\etainv{i'}$ for some 
 star $S_{i'}$ processed earlier or to the facilities in $\etainv{i}$ itself.
Figure \ref{counter} explains the situation. Thus, we need to process the stars in a carefully chosen sequence so as to avoid  this kind of dependency amongst them. That is, we process the stars in a such a way that if we are processing star $\sstar{i}$, then the clients assigned to $i$ in $AS_1$ are not assigned to facilities in $\etainv{i'}$ for a star $S_{i'}$ processed earlier. For this, we construct a weighted directed (dependency) graph $G_2$ on stars.
 We will denote the graph by $G_2 (
\sigma_1, \sigma_2)$ to show that it is a function of the assignments in $AS_1$ and $AS_2$. The graph changes as any of these assignments change.
 
\begin{figure}[]
		\begin{center}
	    \begin{tabular}{c}
		\includegraphics[width=35mm]{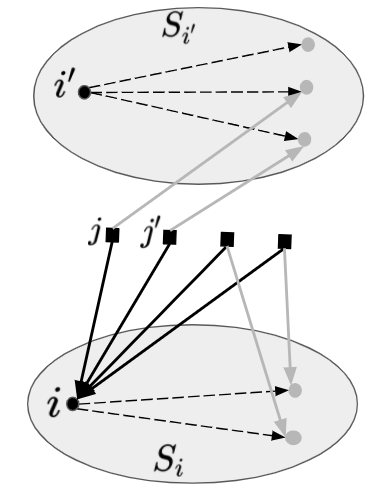}
	    \end{tabular}
		\end{center}
		\caption{Let $\B_i=4$ for all $i \in \facilityset$.  Black and grey edges show the assignment of clients in $AS_1$  and $AS_2$ respectively. Star $S_{i'}$ is processed before star $\sstar{i}$.
		Clients $j, j'$ assigned to $i$ in $AS_1$ have already been assigned to facilities in $\etainv{i'}$ and hence are not available while processing $\sstar{i}$.
		}
	\label{counter}
\end{figure}

The graph $G_2 (\sigma_1, \sigma_2)$ has the stars ($\sstar{i}, i: |\etainv{i}| > 0$) as the vertices and we include the directed edge $(\sstar{i_1}, \sstar{i_2}$) from star $\sstar{i_1}$ to $\sstar{i_2}$ if a client $j \in \sigmaoneinv{i_1}$ is served (in $AS_2$) by some facility $i' \in \eta^{-1}(i_2)$ in star $\sstar{i_2}$ i.e., $|\sigmaoneinv{i_1} \cap (\cup_{i' \in \eta^{-1}(i_{2})}\sigma_2^{-1}(i'))| > 0$.  Let $w(\sstar{i_1}, \sstar{i_{2}}) = |\sigmaoneinv{i_1} \cap (\cup_{i' \in \eta^{-1}(i_{2})} \sigma_2^{-1}(i') )|$ be the weight on the edge $(\sstar{i_1}, \sstar{i_2}$). 
{See Figure \ref{Break_cycles1}.} If the resulting graph is a directed acyclic graph (DAG),  a topological ordering of its vertices gives us a sequence in which the stars can be processed. However, if there are directed cycles in the graph, we redefine the assignments in $AS_1$ to obtain another solution $\hat{AS}_1 = < \facilityset_1, \sigmaprimeone>$ to break the cycles. The dependency graph for $(\sigmaprimeone, \sigma_2)$ will then be an {\em almost-DAG}. We say that a directed graph is an {\em almost-DAG}, if the only directed cycles in it are self loops.

\begin{figure}[]
		\begin{center}
	    \begin{tabular}{c}
		\includegraphics[width=120mm]{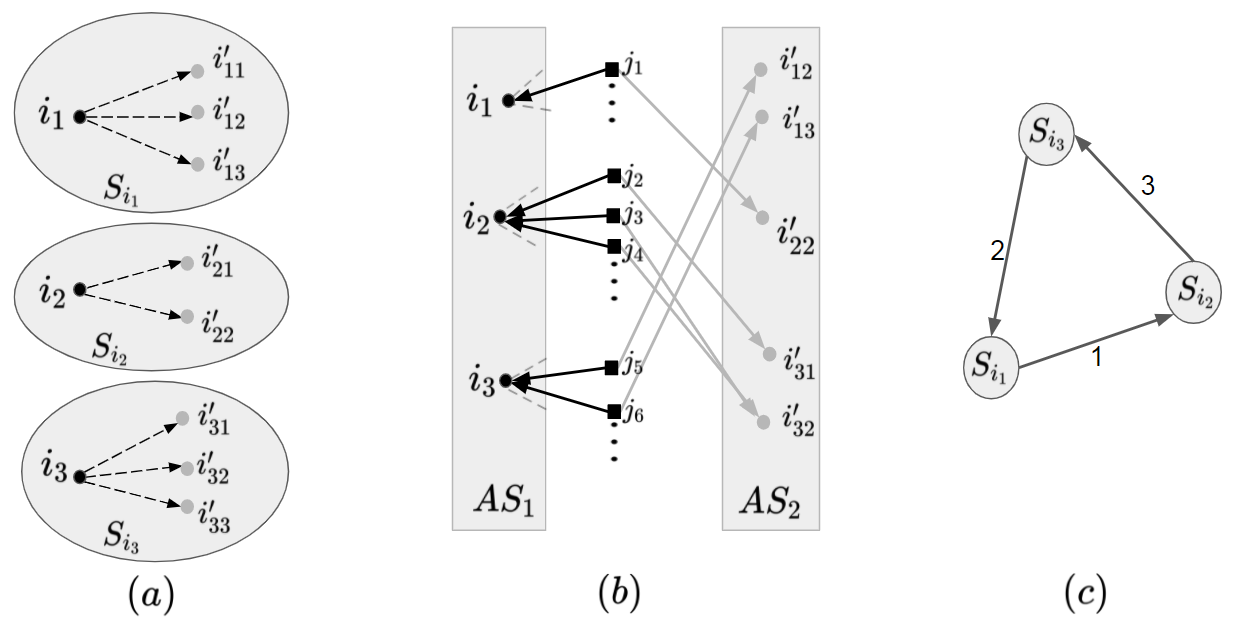}
	    \end{tabular}
		\end{center}
		\caption{($a$) Stars $S_{i_1}$, $S_{i_2}$ and $S_{i_3}$, ($b$) $\sigmaoneinv{i_1} \cap (\cup_{i' \in \eta^{-1}(i_{2})}  \sigma_2^{-1}(i') ) = \{j_1\}$, $\sigmaoneinv{i_2} \cap (\cup_{i' \in \eta^{-1}(i_{3})}  \sigma_2^{-1}(i') ) = \{j_2, j_3,j
		_4\}$, $\sigmaoneinv{i_3} \cap (\cup_{i' \in \eta^{-1}(i_{1})}  \sigma_2^{-1}(i') ) = \{j_5,j_6\}$, ($c$) Directed cycle in  $G_2(\sigma_1, \sigma_2)$.
		}
	\label{Break_cycles1}
\end{figure}

{\bf 
Breaking the cycles} (see Algorithm~\ref{BreakingCycle} and Figure \ref{Break_cycles2}): let $SC = <\sstar{i_1}, \sstar{i_2}, \ldots, \sstar{i_l}>$ be a directed cycle in the graph $G_2(\sigma_1, \sigma_2)$. Wlog, let $(\sstar{i_1}, \sstar{i_2})$ be the minimum weight edge in the cycle. We reassign any $\kappa = w(\sstar{i_1}, \sstar{i_2})$ clients  in $\sigmaoneinv{i_r} \cap (\cup_{i' \in \eta^{-1}(i_{r+1})}\sigma_2^{-1}(i'))$ from $i_r$ to $i_{r+1}$ and reduce the weight of the edge $w(\sstar{i_r}, \sstar{i_{r+1}})$ by $\kappa$ for $r =  1 \ldots l-1$ and reassign any $\kappa$ clients in $ \sigmaoneinv{i_l} \cap (\cup_{i' \in \eta^{-1}(i_1)} \sigma_2^{-1}(i'))$ from $i_l$ to $i_1$ and reduce the weight of the edge $w(\sstar{i_l}, \sstar{i_{1}})$ by $\kappa$. Let $\sigmaprimeone$ denote the new assignments in $AS_1$. Note that $|\sigmaprimeoneinv{i}| = |\sigmaoneinv{i}|$ and hence $|\sigmaprimeoneinv{i}| \ge \B_i$ is maintained for all $i \in \facilityset_1 $ after the reassignments.
The weight of the edge $(\sstar{i_1}, \sstar{i_2})$ becomes zero and we remove it thereby breaking the cycle. 

\begin{figure}[]
		\begin{center}
	    \begin{tabular}{c}
		\includegraphics[width=120mm]{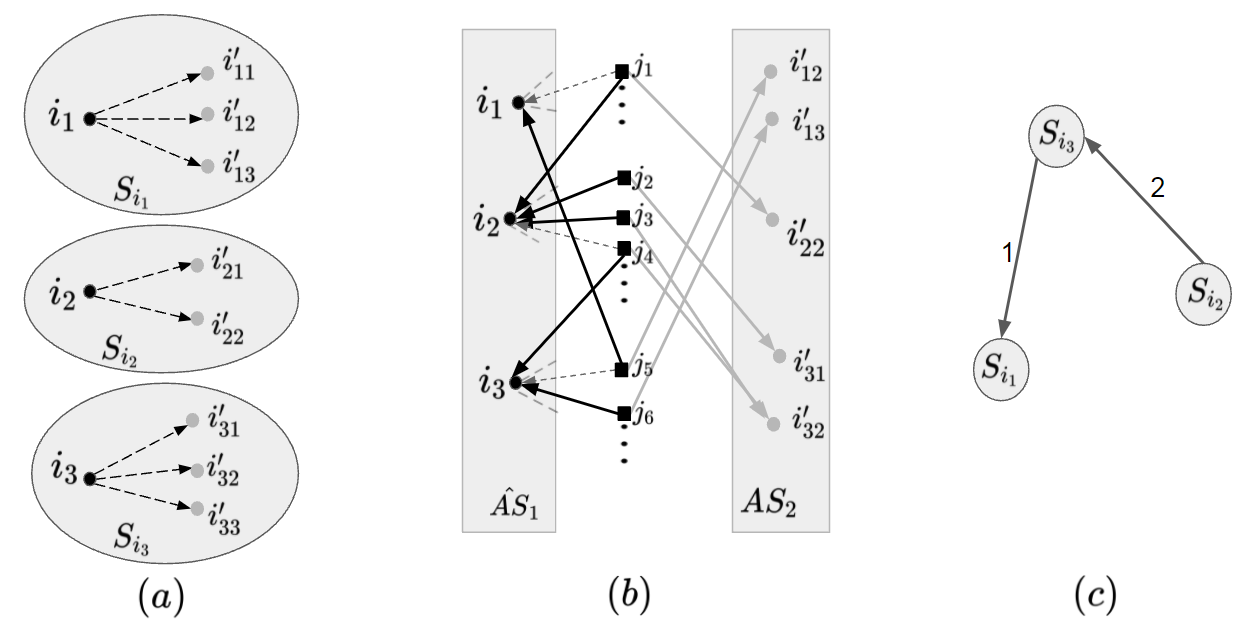}
	    \end{tabular}
		\end{center}
		\caption{Breaking a cycle: ($b$) Assign $j_1$ to $i_2$, $j_4$ to $i_3$ and $j_5$ to $i_1$, that is, $\sigmaprimeone(j_1)=i_2$, $\sigmaprimeone(j_4)=i_3$ and, $\sigmaprimeone(j_5)=i_1$ ($c$) The subgraph of $G_2(\sigmaprimeone, \sigma_2)$ after breaking the cycle.
		}
	\label{Break_cycles2}
\end{figure}

\begin{algorithm}[H] 	
    \setcounter{AlgoLine}{0}
	\SetAlgoLined
    \DontPrintSemicolon  
    \SetKwInOut{Input}{Input}
	\SetKwInOut{Output}{Output}
    \Input{Graph $G_2(\sigma_1, \sigma_2)$}
    
    \Output{$G_2(\sigmaprimeone, \sigma_2)$}
    
    $\sigmaprimeone(j) \gets \sigma_1(j)~\forall j \in \clientset$\;
     
    If $G_2$ is an almost-DAG, return $\sigmaprimeone$.
     
     \For{all directed cycles $SC = <\sstar{i_1}, \sstar{i_2}, \ldots, \sstar{i_l}>$ in $G_2$}
     {
    $\kappa \gets w(\sstar{i_1}, \sstar{i_2})$ \tcp{(Wlog, assume $(\sstar{i_1}, \sstar{i_2})$ as the minimum weight edge in the cycle}
    
    $Bag \gets \emptyset$\;
    
  	 \For {$r = 1$ to $l-1$} {  
		 
		 Add $\kappa$ clients arbitrarily from $\sigmaoneinv{i_r} \cap (\cup_{i' \in \eta^{-1}(i_{r+1})}\sigma_2^{-1}(i'))$ to $Bag$\;
		 
		 \For{$j \in Bag$}{
		 
		 $\sigmaprimeone(j) \gets i_{r+1}$\; 
		 
		 }
		 
		 $Bag \gets \emptyset$\;
		 
		 $w(\sstar{i_r}, \sstar{i_{r+1}}) = w(\sstar{i_r}, \sstar{i_{r+1}}) - \kappa$\;
		 
         }
         
         Add $\kappa$ clients arbitrarily from $\sigmaoneinv{i_l} \cap (\cup_{i' \in \eta^{-1}(i_{1})}\sigma_2^{-1}(i'))$ to $Bag$\;  
         
         \For{$j \in Bag$}{
		 
		 $\sigmaprimeone(j) \gets i_{1}$\; 
		 
		 }
		 
		  $w(\sstar{i_l}, \sstar{i_{1}}) = w(\sstar{i_l}, \sstar{i_{1}}) - \kappa$\;
		 }
	\caption{Breaking Cycles: Constructing an {\em almost} - DAG $G_2(\sigmaprimeone, \sigma_2)$}
	\label{BreakingCycle}
\end{algorithm}

{\bf
Properties of $G_2 (\sigmaprimeone, \sigma_2)$:} ($i$) $G_2 (\sigmaprimeone, \sigma_2)$ is an {\em almost-DAG}, ($ii$) $|\sigmaprimeoneinv{i} |  \ge \B_i$ $\forall$ $i \in \facilityset_1$.

\begin{lemma} \label{ASPcost}
Let $j \in \clientset$. The service cost paid by $j$ in solution $\hat{AS}_1$, $\dist{j}{\sigmaprimeone(j)}$, is bounded by $\dist{j}{\sigma_1(j)} + 2\dist{j}{\sigma_2(j)}$.
\end{lemma}
\begin{proof}
Let $j \in \clientset$. We have $\dist{j}{\sigmaprimeone(j)} \leq  \dist{j}{\sigma_2(j)} + \dist{\sigma_2(j)}{\sigmaprimeone(j)} \leq  \dist{j}{\sigma_2(j)} + \dist{\sigma_2(j)}{\sigma_1(j)} \leq  \dist{j}{\sigma_2(j)} +  (\dist{\sigma_2(j)}{j} + \dist{j}{\sigma_1(j)}) = \dist{j}{\sigma_1(j)} + 2\dist{j}{\sigma_2(j)}$, where the second inequality follows because $\eta(\sigma_2(j)) = \sigmaprimeone(j)$. See Figure \ref{CycleCost}.
\end{proof}

\begin{figure}[]
		\begin{center}
	    \begin{tabular}{c}
		\includegraphics[width=100mm]{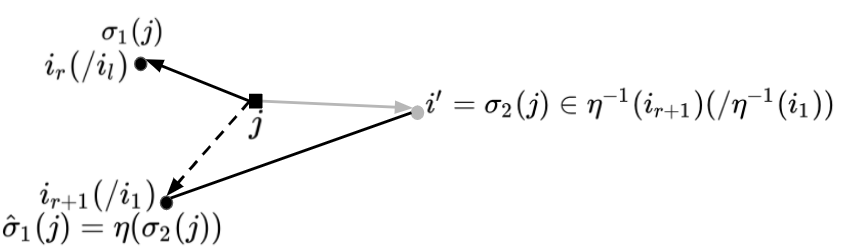}
	    \end{tabular}
		\end{center}
		\caption{ $\dist{j}{\sigmaprimeone(j)} \leq \dist{j}{\sigma_1(j)} + 2\dist{j}{\sigma_2(j)}$. }
	\label{CycleCost}
\end{figure}

Now that we have an {\em almost-DAG} on the stars, we process the stars in the sequence $<S_{i_1}, S_{i_2}, \ldots S_{i_t}>$ defined by a topological ordering of the vertices in $G_2 (\sigmaprimeone, \sigma_2)$ (ignoring the self-loops). While processing the stars, we maintain partition of our clients into two sets $\clientset_s$ and $\clientset_u$ of {\em settled} and {\em unsettled} clients respectively. We say that a client is {\em settled} if it has been assigned to an open facility in $S_I$ and {\em unsettled} otherwise. Initially $\settled = \emptyset$ and $\unsettled = \clientset$. As we process the stars, more and more clients get settled.

Consider star $\sstar{i}$. Algorithm \ref{TreeProcessing}
gives the  processing of $\sstar{i}$ in detail. For $i' \in \etainv{i}$, let $N_{i'}$ be the set of unsettled clients, assigned to $i'$ in $AS_2$. We open some facilities in $\etainv{i}$ and assign these clients to them. Consider the facilities in $\etainv{i}$ in decreasing order of distance from $i$, i.e., $y_1,y_2,...,y_l$. 
 We collect the unsettled clients assigned to them, by $AS_2$, in a bag looking for a facility $t$ at which we have collected at least $\B_t$ clients. We open such a facility and empty the bag by assigning all the clients in the bag to the facility just opened ({\em Type-I assignment}) and start the process again with the next facility in the order (lines $10 - 16$). 
 To avoid the problem mentioned earlier and exhibited in Figure $\ref{counter}$, 
 before we do this,
 we reserve some clients in $\sigmaprimeoneinv{i}$
 so that the total number of unsettled clients assigned (in $AS_2$) to the last facility $y_l$ in the order, along with the reserved clients is at least $\B_i$. For this, we reserve $\max\{0, \B_i - |N_{y_l}|\}$ clients from $\sigmaprimeoneinv{i}\setminus N_{y_l}$  at $i$ (line $7$). 
 Note that 
 $\sigmaprimeoneinv{i}$ and hence $\sigmaprimeoneinv{i}\setminus N_{y_l}$ may contain clients from $\cup_{i' \in \etainv{i}} N_{i'}$. Thus, we delete the reserved clients from
 $N_{i'},~i' \in \etainv{i}$ (lines $8 - 9$) before processing the facilities in $\etainv{i}$. 
 
 To make sure that we do not open more than $|\etainv{i}|$ facilities in $\sstar{i}$, we open only one of $i$ and $y_l$ for the remaining $(Bag \cup N_{y_{l}} \cup \res{i})$ clients. This also ensures that we do not open $i$ more than once. We open $i$ and give all the remaining clients 
 to $i$ because (as we will show later) the cost of assigning clients from $Bag \cup N_{y_{l}}$ to $i$ is bounded whereas we do not know how to bound the cost of assigning clients in $\res{i}$ to $y_l$.
 Clearly, we open no more than $k$ facilities in all.

\begin{claim}
At line $7$, we have sufficient clients in $\sigmaprimeoneinv{i} \setminus N_{y_l}$ i.e., $|\sigmaprimeoneinv{i} \setminus N_{y_l}| \ge \B_i - |N_{y_l}|$.
\end{claim}

\begin{proof}
Note that the set of clients that get settled while processing star $\sstar{i}$ is a subset of $\sigmaprimeoneinv{i} \cup (\cup_{i' \in \eta^{-1}(i)} \sigmatwoinv{i'})$.
  To prove the claim, we only need to prove that $\sigmaprimeoneinv{i} \cap \settled = \emptyset$. The claim then follows by using the fact that
$|\sigmaprimeoneinv{i}| \ge \B_i$.
Let $\sstar{\bar{i}}$  be a star that was processed before the star $\sstar{i}, \bar i \ne i$. Then since $G_2(\sigmaprimeone, \sigma_2)$ is an {\em almost-DAG}, there is no edge from $\sstar{i}$ to $\sstar{\bar i}$, i.e., $|\sigmaprimeoneinv{i} \cap (\cup_{i' \in \eta^{-1}(\bar i)}\sigma_2^{-1}(i'))| = 0$. Obviously, $|\sigmaprimeoneinv{i} \cap \sigmaprimeoneinv{\bar{i}} | = 0$. Thus, the set of clients that got settled while processing $\sstar{\bar{i}}$ has no intersection with $\sigmaprimeoneinv{i}$. Hence the claim follows.
\end{proof}

\begin{algorithm}[H] 	
    \setcounter{AlgoLine}{0}
	\SetAlgoLined
    \DontPrintSemicolon  
    \SetKwInOut{Input}{Input}
	\SetKwInOut{Output}{Output}
    \Input{$\sstar{i}: i \in \facilityset_1$}
    
    $\res{i} \gets \emptyset$\;
    
    $Bag \gets \emptyset$\;

\For {$i' \in \etainv{i}$} {$N_{i'} \gets \clientset_u \cap \sigmatwoinv{i'}$ }

	Arrange the facilities in $\eta^{-1}(i)$ in the sequence $<y_1, \ldots y_l>$ such that $\dist{y_{l'}}{i} \geq \dist{y_{l' + 1}}{i} ~\forall ~l' = 1 \ldots l - 1$\;
	
	\If{$|N_{y_l}| < \B_i$}
	{
	$\res{i} \gets$ set of any $\B_i - |N_{y_l}|$ clients from $\sigmaprimeoneinv{i} \setminus N_{y_l}$

	 \For {$i' \in \etainv{i}$}{
        $N_{i'} \gets N_{i'} \setminus \res{i} $
    
    }
	}

		 \For {$l' = 1$ to $l-1$} {  
		 $Bag \gets Bag \cup N_{y_{l'}}$\;
		 
		     \If{$|Bag| \geq \B_{y_{l'}}$}  {
		 
		        Open facility $y_{l'}$\;
	       
		       \For{$j \in Bag$}{
                    Assign $j$ to $y_{l'}$, $\settled \gets \settled \cup \{ j \}$, $\unsettled \gets \unsettled \setminus \{j \}$\;
                }
                
                $Bag \gets \emptyset$\;
		    }
		    }
		    
		    
		    $t\gets i$\;
		    
		    If $|Bag \cup N_{y_{l}} \cup \res{i}| > (\beta+1)\capacityU_i$ then $t \gets y_l$

		    Open $t$\;
		    
		    \For{$j \in Bag \cup N_{y_{l}} \cup \res{i}$}{ 
            
            Assign $j$ to $t$, $\settled \gets \settled \cup \{ j \}$, $\unsettled \gets \unsettled \setminus \{j \}$\;
            }
            
	\caption{Process($\sstar{i}$)}
	\label{TreeProcessing}
\end{algorithm}

Note that when one of the bounds is uniform, we have $\capacityU_i \ge \B_{i'}$ for all $i$ and $i'$.
Clearly, the facilities opened by the above algorithm in line $13$ ({\em Type-I} assignment) satisfy the lower bounds. 
Let the assignment of clients to facility $i$ when $t=i$ in lines $20$ - $21$ be called as {\em Type-II assignments} and those to facility $y_l$ when $t=y_l$ be called as {\em Type-III assignments}.
In {\em Type-II} assignments, the star center $i$ satisfies the lower bound (if opened at at line $19$) as $|Bag \cup N_{y_l} \cup \res{i}| \geq \B_i$ where the inequality follows because $|\res{i}| = \max\{0, \B_i -  N_{y_l}\}$. In {\em Type-III} assignments, facility $y_l$ (if opened at line $19$) also satisfies the lower bound as $|Bag \cup N_{y_l} \cup \res{i}| > (\beta+1) \capacityU_i 
\geq 2 \B_{y_l}$  when one of the bounds is uniform.

We next bound the violations in the upper bound.
Consider the facilities in $\eta^{-1}(i)$. These facilities receive clients only in {\em Type-I} assignments  (lines $14-15$). Note that for $l' = 2,...,l-1$, we have $|Bag| < \B_{y_{l'-1}}$ just before line $11$ and hence $|Bag| < \B_{y_{l' - 1}} +  \beta \capacityU_{y_{l'}}$ (just after line $11$) $\leq (1+\beta)\capacityU_{y_{l'}}$  when one of the bounds is uniform. For $l'=1$, $|Bag| = 0$ just before line $11$ and hence $|Bag| \leq \beta \capacityU_{y_1}$ (just after line $11$). 
For {\em Type-II} assignments, the bound holds trivially because the center $i$ receives clients only when $|Bag| + |N_{y_l}| + |\res{i}| \leq (\beta + 1)\capacityU_i$. 
The maximum number  of clients received by facility $y_{l}$ in {\em Type-III} assignments is no more than $|Bag| + |N_{y_l}| + |\res{i}| = |Bag| + |N_{y_l}| + max\{0, \B_i - N_{y_l}\} =  |Bag| + max\{\B_i, N_{y_l}\} \le \B_{y_{l-1}} + \beta \capacityU_{y_{l}} \le (\beta+1) \capacityU_{y_{l}}$ when either of the bounds is uniform. 


Next, we bound the service cost. 
Consider a star $\sstar{i}$.
Consider a client $j$ assigned to a  facility $i_2 \in \etainv{i}$ in {\em Type-I} assignments  where  
$j$ was assigned to $i_1\in \etainv{i}$ in $AS_2$
i.e., $i_1 = \sigma_2(j)$ and $i_2 = \sigma_I(j)$.  See Figure \ref{TypeI/II}-(a). The cost paid by $j$ is $\dist{i_2}{j} \leq \dist{i_1}{j} + \dist{i_1}{i} + \dist{i}{i_2} \leq \dist{i_1}{j} + 2\dist{i_1}{i} \leq \dist{i_1}{j} + 2\dist{i_1}{\sigmaprimeone(j)} \leq \dist{i_1}{j} + 2(\dist{i_1}{j} + \dist{j}{\sigmaprimeone(j)}) = 3\dist{i_1}{j} + 2\dist{j}{\sigmaprimeone(j)} = 3\dist{j}{\sigma_2(j)} + 2\dist{j}{\sigmaprimeone(j)} \leq 2\dist{j}{\sigma_1(j)} + 7\dist{j}{\sigma_2(j)}$ where the third inequality follows because $\eta(i_1) = i$ and and last follows by Lemma \ref{ASPcost}. Next, consider {\em Type-II} assignments. 
Let $j \in \res{i}$ be assigned to $i$. Then since $\sigmaprimeone(j) = i$, the service cost $\dist{i}{j} (= \dist{\sigmaprimeone(j)}{j})$ is bounded by $\dist{j}{\sigma_1(j)} + 2\dist{j}{\sigma_2(j)}$ by Lemma~\ref{ASPcost}.
Next, let $j \in N_{i'}:~i' \in \etainv{i}$ be a client assigned  to $i$ i.e. $i' = \sigma_2(j)$ and $i = \sigma_I(j)$.
See Figure \ref{TypeI/II}-(b). Then, the service cost $\dist{i}{j}$ is bounded by
$\dist{\sigma_2(j)}{j} + \dist{\sigma_2(j)}{i} = \dist{\sigma_2(j)}{j} + \dist{\sigma_2(j)}{\eta(\sigma_2(j))} \leq \dist{\sigma_2(j)}{j} + \dist{\sigma_2(j)}{\sigmaprimeone(j)} \leq \dist{\sigma_2(j)}{j} +  \dist{\sigma_2(j)}{j} + \dist{j}{\sigmaprimeone(j)} = 2\dist{j}{\sigma_2(j)} + \dist{j}{\sigmaprimeone(j)} \leq 4\dist{j}{\sigma_2(j)} + \dist{j}{\sigma_1(j)}$
where the last inequality follows by Lemma \ref{ASPcost}. 
For {\em Type-III} assignments, note that $|Bag \cup N_{y_l} \cup \res{i}| > (\beta+1)\capacityU_i \Rightarrow |\res{i}| = 0$, for otherwise $|N_{y_l} \cup \res{i}| = \B_{i}$ and thus $|Bag \cup N_{y_l} \cup \res{i}| < \B_{Prev} + \B_{i} \le 2\capacityU_i$ when either of the bounds is uniform. Hence, the cost of assigning $|Bag \cup N_{y_l}|$ clients to $y_l$ is bounded in the same manner as the cost of {\em Type-I} assignments.

\begin{figure}[]
		\begin{center}
		\includegraphics[width=115mm]{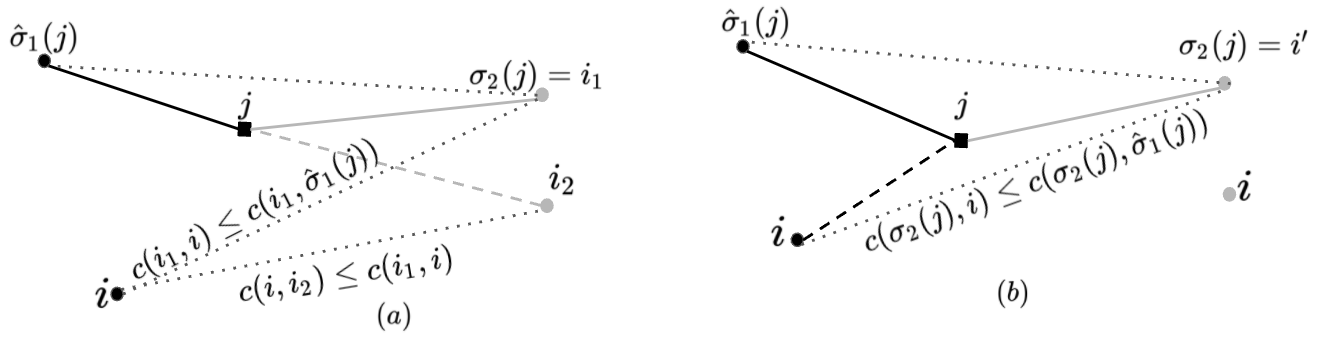}
	    \end{center}
		\caption{($a$)
		Cost bound of {\em Type-I} assignments: 
		($b$) Cost bound of {\em Type-II} assignments: 
		}
	\label{TypeI/II}
\end{figure}

Let $FCost_{\mathcal{I}}(S)$ and $SCost_{\mathcal{I}}(S)$ denote the facility opening cost and service cost of a solution $\mathcal{S}$ of instance $I$ respectively. Then, $FCost_{\mathcal{I}}(S) \le FCost_{\mathcal{I}_1}(AS_1) + FCost_{\mathcal{I}_2}(AS_2)$ and  $SCost_{\mathcal{I}}(S) \le 2 SCost_{\mathcal{I}_1}(AS_1) + 7 SCost_{\mathcal{I}_2}(AS_2)$.


\noindent {\bf Proof of Theorem~\ref{LBUBkFL1Theorem1}:} Using a $\gamma$-approximation for instance $I_1$ of LFL or LkS/LkC, as the case may be, with uniform lower bounds and an $\alpha$-approximation for instance $I_2$ of  UkM/UkFL/UFL/UkS/UkC, we get the desired claims.

\noindent {\bf Proof of Theorem~\ref{LBUBkFLTheorem}:} Using a $\gamma$-approximation for instance $I_1$ of LFL or LkS/LkC, as the case may be, and an $\alpha$-approximation for instance $I_2$ of  UkM/UkFL/UFL/UkS/UkC  with uniform upper bounds, we get the desired claims.

Algorithm \ref{Step5} summarizes Step $5$ of Algorithm \ref{LBUBkM_Algo}.

\begin{algorithm}[H] 	
	\label{Step5}
    \setcounter{AlgoLine}{0}
	\SetAlgoLined
    \DontPrintSemicolon  
    \SetKwInOut{Input}{Input}
	\SetKwInOut{Output}{Output}
    \Input{$<AS_1=\facilityset_1,\sigma_1>$, $<AS_2=\facilityset_2,\sigma_2>$}
    \Output{$AS_I$}

    Construct graph $G_1 = <\facilityset_1 \cup \facilityset_2, E>$ where $E = \{(i', \eta(i')): i' \in \facilityset_2\}$\;
    
    
    Construct an {\em almost-DAG} $G_2 (\sigmaprimeone, \sigma_2)$ from $G_1$  using algorithm \ref{BreakingCycle} 
    
    
    Obtain a topological ordering $<\sstar{i_1}, \sstar{i_2} \ldots \sstar{i_t}>$ of stars in the {\em almost-DAG} $G_2 (\sigmaprimeone, \sigma_2)$.
    
    \For{$r = 1$ to $t$}{
    Process $\sstar{i_r}$ using Algorithm \ref{TreeProcessing}\;
    }
   
	\caption{Constructing $AS_I$}
	\label{Step5}
\end{algorithm}

\section{Reducing violation in upper bounds when $2\B_t \leq \capacityU_t$} \label{2LU}

In this section, assuming $2\B_t \leq \capacityU_t~\forall t \in \facilityset$, we modify Algorithm $\ref{TreeProcessing}$ to obtain Algorithm $\ref{TreeProcessing2L}$ that 
reduces the violation in upper bounds from $(\beta + 1)$ to $(\beta + \epsilon)$ for a given $\epsilon > 0$ when one of the bounds is uniform.
We do the following modifications to the algorithm: $(i)$ on arriving at a facility, say $t$, at which $|Bag| \ge \B_t$, we open $t$ and instead of emptying the bag, we assign only $\beta \capacityU_t$ clients to $t$. Remaining clients are carried forward to the next facility in the order. $(ii)$ We keep account of the last facility (in {\em  Prev}), if any, that is not opened, and the number of clients in the bag at that instant (in {\em  $Prev_{count}$}) i.e., {\em Prev} is the facility $y_{l'}$ for which $|Bag| < \B_{y_{l'}}$ immediately after line $11$  (hence at line $18$) and {\em  $Prev_{count}$} $= |Bag|$ at that time. We open {\em Prev}  at the end, if required. This is done as follows: if $|Bag \cup \res{i} \cup N_{y_l}| \le (\beta + \epsilon) \capacityU_i(/\capacityU_{y_l})$, we are done 
(we open $i(/y_l)$ and assign all clients to it). Else, we open both {\em Prev} and $y_l$ (at line $31$) (note that
{\em Prev} $\ne y_l$ must exist in this case) and, assign the clients in $Bag \cup \res{i} \cup N_{y_l}$ to {\em Prev} and $y_l$, so that they receive at least $\B_{Prev}$ and $\B_{y_l}$ clients respectively. 
When $2 \B_t \le \capacityU_t$ and  at least one of the bounds is uniform, it is possible to do so
as $\capacityU_t \ge \B_r + \B_s$ for all $r, s$ and $t$.
We will show that the violation in upper bound and the service cost are bounded in this case. 

\begin{algorithm}[H] 
    \setcounter{AlgoLine}{0}
	\SetAlgoLined
    \DontPrintSemicolon  
    \SetKwInOut{Input}{Input} 
	\SetKwInOut{Output}{Output}
    \Input{$\sstar{i}: i \in \facilityset_1$}
    
    $\res{i} \gets \emptyset$, $Bag \gets \phi$\;

\For {$i' \in \etainv{i}$} {$N_{i'} \gets \clientset_u \cap \sigmatwoinv{i'}$ }

	Arrange the facilities in $\eta^{-1}(i)$ in the sequence $<y_1, \ldots y_l>$ such that $\dist{y_{l'}}{i} \geq \dist{y_{l' + 1}}{i} ~\forall ~l' = 1 \ldots l - 1$\;
	
	\If{$|N_{y_l}| < \B_i$}
	{
	$\res{i} \gets$ set of any $\B_i - |N_{y_l}|$ clients from $\sigmaprimeoneinv{i} \setminus N_{y_l}$

	 \For {$i' \in \etainv{i}$}{
        $N_{i'} \gets N_{i'} \setminus \res{i} $
    
    }
	}
        
        $Prev \gets null$, 
        {$Prev_{count} =0$}\;
        
		 \For {$l' = 1$ to $l-1$} {  
		 $Bag \gets Bag \cup N_{y_{l'}}$\;
		 
		     \eIf{$|Bag| \geq \B_{y_{l'}}$}  {
		 
		        Open facility $y_{l'}$\;
		        
	            $Count \gets 0$\;
	            
		       \For{$j \in Bag$}{
		            \If{$Count\leq \beta\capacityU_{y_{l'}}$}
		            {
		            
                    Assign $j$ to $y_{l'}$, $\settled \gets \settled \cup \{ j \}$, $\unsettled \gets \unsettled \setminus \{j \}$, $Bag \gets Bag \setminus \{ j\}$, $Count++$,\;
                }
                }
		    }
		    {
		    $Prev \gets y_{l'}$\;  \tcp{$Prev$ denotes the last unopened facility in $\etainv{i}$}
		    
		    $Prev_{count} = |Bag|$\;
		    }
		    }
		    
		    \If{$|Bag \cup N_{y_{l}} \cup \res{i}| \leq (\beta+\epsilon) \capacityU_i$}{
		    
		    Open $i$\;
		    
		    \For{$j \in Bag \cup N_{y_{l}} \cup \res{i}$}{
		    
		    Assign $j$ to $i$, $\settled \gets \settled \cup \{ j \}$, $\unsettled \gets \unsettled \setminus \{j \}$\;
		    }
		    
		    return \;
		    }
		    \If{$|Bag \cup N_{y_{l}} \cup \res{i}| \leq (\beta+\epsilon) \capacityU_{y_l}$}{
		    
		   Open $y_l$\;
		   
		   \For{$j \in Bag \cup N_{y_{l}} \cup \res{i}$}{
		    
		    Assign $j$ to $y_l$, $\settled \gets \settled \cup \{ j \}$, $\unsettled \gets \unsettled \setminus \{j \}$\;
		    }
		    
		    return \;
		    }

		    Open $Prev$ and $y_l$\; \tcp{$\res{i} =0$ when $|Bag \cup N_{y_{l}} \cup \res{i}| > (\beta+\epsilon) \capacityU_{y_l}(/\capacityU_i)$}
		    
		    $Count \gets 0$\;
		    
		    $Bag \gets Bag \cup N_{y_{l}} \cup \res{i}$\;
		    
		   \For{$j \in Bag$}{
		  
            \eIf{$Count \leq \B_{Prev}$}{
                
                Assign $j$ to $Prev$, $\settled \gets \settled \cup \{ j \}$, $\unsettled \gets \unsettled \setminus \{j \}$, $Bag \gets Bag \setminus \{ j \}$, $Count++$\;
                    }
            {
            Assign all remaining clients in $Bag$ to $y_l$ and Break 
	            }
	       }
       
	\caption{Process($\sstar{i}$) when $2\B_t \leq \capacityU_t~\forall t \in \facilityset$}
	\label{TreeProcessing2L}
\end{algorithm}


Let the assignment of clients to facility {\em Prev} in line $36$ be called as {\em Type-IV assignments}. 
The assignments in line $17$, line $24$ and lines $29~\&~38$ are {\em Type-I}, {\em Type-II} and {\em Type-III} assignments respectively. 
Before we proceed to prove our claims, note that  we open at most one of $y_l$ and $i$: if $i$ is opened at line $22$, we return at line $25$ and thus $y_l$ is never opened in this case.  As before, this ensures that $i$ is not opened more than once.

Clearly, lower bound is satisfied by the  {\em Type-I} and {\em Type-IV} assignments done in line $17$ and $36$ for the facilities opened in lines $13$ and $31$ respectively.
Also, since $|Bag \cup N_{y_l} \cup \res{i}| \geq \B_i$, 
lower bound is satisfied by the  {\em Type-II} assignments done in line $24$ for the facility $i$ opened in line $22$. At line $29$,  $|Bag \cup N_{y_l} \cup \res{i}| > (\beta+\epsilon) \capacityU_i \geq (\beta+\epsilon) \B_{y_l}$ when either of the bounds is uniform,
Clearly, the upper bound is violated at most by a factor of $(\beta + \epsilon)$ at all the facilities opened in lines $17$, $24$, $29$ and $36$.
To bound the assignments done in line $38$, we look at the status at line $31$.
At line $31$, $|\res{i}| = 0$, for otherwise $|N_{y_l} \cup \res{i}| = \B_{i}$, hence $|Bag \cup \res{i} \cup N_{y_l}| < \B_{Prev} + \B_{i} \le \capacityU_i$ (when one of the bounds is uniform). Thus, $|Bag \cup N_{y_l} \cup \res{i}| = |Bag \cup N_{y_l}| < \B_{Prev} + \beta \capacityU_{y_l}$. Also, $|Bag \cup N_{y_l} \cup \res{i}|  >  (\beta + \epsilon) \capacityU_{y_l} > \B_{Prev} + \B_{y_l}$ (when one of the bounds is uniform). Thus, $\B_{y_l} < |Bag \cup N_{y_l} \cup \res{i}| - \B_{Prev} < \beta \capacityU_{y_l}$ i.e., at line $38$, $\B_{y_l} < |Bag| < \beta \capacityU_{y_l}$.



Costs of {\em Type-I} assignments, {\em Type-II} assignments and {\em Type-III} assignments are bounded in the same manner as 
in Section~\ref{bicriteriakM}.
To bound the service cost of {\em Type-IV} assignments (line $36$),
observe that $|Bag \cup \res{i} \cup N_{y_l}|  {> (\beta + \epsilon) \capacityU_{y_l}}$ implies that $|Bag| > \epsilon  \capacityU_{y_l}$  as $|\res{i}| = 0$ and $N_{y_l} \le \beta \capacityU_{y_l}$; hence ${Prev_{count}} \ge |Bag| ($at line$ ~20 )> \epsilon  \capacityU_{y_l} >  \epsilon \B_{Prev}$ (last inequality holds when at least one of the bounds is uniform).
Note that $Prev$ and $Prev_{count}$ do not change after exiting the for-loop at line $20$. Thus,  $Prev_{count} > \epsilon \B_{Prev}$ after line $31$ also. 
Thus, the cost of assigning at most $\B_{Prev}$ clients from $N_{y_l}$ to {\em Prev} is bounded by $(1/\epsilon)$ times the cost of assigning $\epsilon \B_{Prev}$ clients from $\cup_{ \text{$i'$ occurs before {\em Prev} in the order}} N_{i'} \cup N_{{Prev}}$ to  $y_l$. Hence, the total cost of {\em Type-IV} assignments is bounded by $(1/\epsilon)$ total cost of {\em Type-III} assignments). 

Theorem \ref{FinalResultLBUBkFL2L11} and  \ref{FinalResultLBUBkFL2L1} then follow in the same manner as Theorem \ref{LBUBkFL1Theorem1} and  \ref{LBUBkFLTheorem} respectively.

\section{Conclusion and Future Work}
\label{cncl}

 In this paper, we presented first constant factor approximations for lower and upper bounded $k$-median and $k$-facility location problems violating the upper bounds by a factor of $\beta+1$ where $\beta$ is the violation in upper bounds in the solutions of the underlying problems with upper bounds.  We studied the problems when one of the bounds is uniform. Any improvement in $\beta$ in future will lead to improved results for our problems as well. Our approach also gives a constant factor approximation for lower and upper bounded facility location; we improve upon the upper bound violation of $5/2$ obtained by Gupta \etal~\cite{GroverGD21_LBUBFL_Cocoon} to $2$. We also presented first constant factor approximations for lower and  uniform upper bounded $k$-center and its generalization, $k$-supplier problem. 

 
For the special case when $2\B_i \le \capacityU_i~\forall i \in \facilityset$, an improvement in upper bound violation to $\beta+\epsilon$ was also obtained  for a given $\epsilon>0$.

 \noindent {\bf Additional results: } Our 
framework also provides a polynomial time algorithm that approximates LUkFL within a constant factor violating the upper bounds by a factor of $(2 + \epsilon)$ and cardinality by a factor of $\frac{2}{1 + \epsilon}$ for a given $\epsilon > 0$  using approximation of Grover \etal~\cite{GroverGKP18} that violates capacities by $(1 + \epsilon)$ factor and cardinality by $(\frac{2}{1 + \epsilon})$. We also get a result that violates the upper bound by $2$ factor and cardinality by a factor of $2$ when the facility costs are uniform by using Aardal \etal~\cite{capkmGijswijtL2013}. 

\noindent \textbf{Future Work:} One direction for future work would be to obtain true constant factor approximations for the problems. For kM and kFL, this would be challenging as no true approximations are known for the upper bounded variants of the underlying problems. Can we obtain a true approximation for LUFL?
For LUkS/LUkC, true approximations are known for uniform lower bounds and  general upper bounds. However, to obtain a true approximation with general lower bounds and uniform upper bounds is open for the problem.

\bibliography{ref_master_CFL}
\end{document}